\documentclass{article}
\usepackage[totalwidth=13.0cm,totalheight=20.0cm]{geometry}

\usepackage{latexsym,amsthm,amsmath,amssymb,url}

\newtheorem{lemma}{Lemma}
\newtheorem{corollary}{Corollary}

\newtheorem{theorem}{Theorem}
\newtheorem{remark}{Remark}
\newtheorem{proposition}{Proposition}

\newcommand{\MLT}{\textsc{Max Lin AA}}
\newtheorem{krule}{Reduction Rule}

\begin{document}

\title{Systems of Linear Equations over $\mathbb{F}_2$ and Problems Parameterized Above Average}

\author{
R. Crowston, G. Gutin, M. Jones, E.J. Kim\\
{\small Royal Holloway, University of London}\\[-3pt]
{\small Egham, Surrey, TW20 0EX, UK}\\[-3pt]
{\small \url{{robert|gutin|markj|eunjung}@cs.rhul.ac.uk}}\\
I.Z. Ruzsa\\
{\small Alfr{\'e}d R{\'e}nyi Institute of Mathematics, Hungarian Academy of Sciences}\\[-3pt]
{\small H-1053, Budapest, Hungary, \url{ruzsa@renyi.hu}}
}
\date{}
\maketitle

\begin{abstract}
\noindent
In the problem Max Lin, we are given a system $Az=b$ of $m$
linear equations with $n$ variables
over $\mathbb{F}_2$ in which each equation is assigned a positive weight
and we wish to find an assignment of
values to the variables that maximizes the excess, which is the total weight of
satisfied equations minus the total weight of falsified equations. Using an algebraic
approach, we obtain a lower bound for the maximum excess.

Max Lin Above Average (Max Lin AA) is a parameterized version of Max Lin
introduced by Mahajan et al. (Proc. IWPEC'06 and J. Comput. Syst. Sci.  75, 2009). In Max Lin AA
all weights are integral and we are to
decide whether the maximum excess is at least
$k$, where $k$ is the parameter.

It is not hard to see that we may assume that no two equations in $Az=b$ have the same left-hand side
and $n={\rm rank A}$. Using our maximum excess results, we prove that, under these assumptions, Max Lin AA is fixed-parameter tractable for a wide special case: $m\le 2^{p(n)}$ for an arbitrary fixed function $p(n)=o(n)$. This result generalizes earlier results by Crowston et al. (arXiv:0911.5384) and Gutin et al. (Proc. IWPEC'09). We also prove that Max Lin AA is polynomial-time solvable for every fixed $k$ and, moreover, Max Lin AA is in the parameterized complexity class W[P].

Max $r$-Lin AA is a special case of Max Lin AA, where each equation has at most $r$ variables. In Max Exact $r$-SAT AA we are given a multiset of $m$ clauses on $n$ variables such that each clause has $r$ variables and asked whether there is a truth assignment to the $n$ variables that
satisfies at least $(1-2^{-r})m + k2^{-r}$ clauses.
Using our maximum excess results, we prove that for each fixed $r\ge 2$, Max $r$-Lin AA and Max Exact $r$-SAT AA can be solved in time $2^{O(k \log k)}+m^{O(1)}.$ This improves $2^{O(k^2)}+m^{O(1)}$-time algorithms for the two problems obtained by Gutin et al. (IWPEC 2009) and Alon et al. (SODA 2010), respectively.

It is easy to see that maximization of arbitrary pseudo-boolean functions, i.e., functions $f:\ \{-1,+1\}^n\rightarrow \mathbb{R}$, represented by their Fourier expansions
is equivalent to solving Max Lin.  Using our main maximum excess result,
we obtain a tight lower bound on the maxima of pseudo-boolean functions.
\end{abstract}

\section{Introduction}\label{section:intro}

In the problem \textsc{Max Lin}, we are given a system $Az=b$ of $m$
linear equations in $n$ variables
over $\mathbb{F}_2$ in which each equation is assigned a positive weight
and we wish to find an assignment of
values to the variables in order to maximize the total weight of
satisfied equations. A special case of \textsc{Max Lin} when each equation  has at most $r$ variables is called {\sc Max $r$-Lin}.

Various algorithmic aspects of \textsc{Max Lin} have been well-studied
(cf. \cite{AGK04,Hastad,HastadVenkatesh02}). Perhaps, the best known result on
\textsc{Max Lin}
is the following inapproximability theorem of H{\aa}stad \cite{Hastad}:
 unless P=NP, for each $\epsilon >0$ there is no polynomial time algorithm for
distinguishing instances of \textsc{Max 3-Lin} in which at least
$(1-\epsilon)m$ equations can be simultaneously satisfied from instances
in which less than $(1/2 + \epsilon)m$ equations can be simultaneously satisfied.

Notice that maximizing the total weight of satisfied equations is equivalent to maximizing the {\em excess},
which is the total weight of satisfied equations minus the total weight of falsified equations. In Section \ref{CSPsec}, we investigate lower bounds for the maximum excess. Using an algebraic approach, we prove the following main result: Let $Az=b$ be a \textsc{Max Lin} system such that ${\rm rank} A = n$ and no pair of equations has the same left-hand side, let $w_{\min}$ be the minimum weight of an equation in $Az=b$, and let $k\ge 2.$ If $k\le m\le 2^{n/(k-1)}-2$, then the maximum excess of $Az=b$ is at least $k\cdot w_{\min}$. Moreover, we can find an assignment that achieves an excess of at least $k\cdot w_{\min}$ in time $m^{O(1)}$.

Using this and other results of Section \ref{CSPsec} we prove parameterized complexity results of Section \ref{LBsec}.
To describe these results we need the following notions, most of which can be found in monographs \cite{DowneyFellows99,FlumGrohe06,Niedermeier06}.

A \emph{parameterized problem} is a subset $L\subseteq \Sigma^* \times
\mathbb{N}$ over a finite alphabet $\Sigma$. $L$ is
\emph{fixed-parameter tractable} if the membership of an instance
$(x,k)$ in $\Sigma^* \times \mathbb{N}$ can be decided in time
$f(k)|x|^{O(1)},$ where $f$ is a computable function of the
parameter $k$.
When the decision time is replaced by the much more powerful $|x|^{O(f(k))},$
we obtain the class XP, where each problem is polynomial-time solvable
for any fixed value of $k.$ There is an infinite number of parameterized complexity
classes between FPT and XP (for each integer $t\ge 1$, there is a class W[$t$]) and they form the following tower:
$$FPT \subseteq W[1] \subseteq W[2] \subseteq \cdots \subseteq W[P] \subseteq XP.$$
Here W[P] is the class of all parameterized problems $(x,k)$ that can be decided in $f(k)|x|^{O(1)}$ time
by a nondeterministic Turing machine that makes at most $f(k)\log |x|$ nondeterministic steps for some computable function $f$.
For the definition of classes W[$t$],
see, e.g., \cite{FlumGrohe06} (we do not use these classes in the rest of the paper).

Given a pair $L,L'$ of parameterized problems,
a \emph{bikernelization from $L$ to $L'$} is a polynomial-time
algorithm that maps an instance $(x,k)$ to an instance $(x',k')$ (the
\emph{bikernel}) such that (i)~$(x,k)\in L$ if and only if
$(x',k')\in L'$, (ii)~ $k'\leq f(k)$, and (iii)~$|x'|\leq g(k)$ for some
functions $f$ and $g$. The function $g(k)$ is called the {\em size} of the bikernel.
The notion of a bikernelization was introduced in \cite{AlonSODA2010}, where it was observed that
a parameterized problem $L$ is fixed-parameter
tractable if and only if it is decidable and admits a
bikernelization from itself to a parameterized problem $L'$.
A {\em kernelization} of a parameterized problem
$L$ is simply a bikernelization from $L$ to itself; the bikernel is the {\em kernel}, and $g(k)$ is the {\em size} of the kernel.
Due to applications, low degree polynomial size kernels are of main interest.

Note that $W/2$ is a tight lower bound on the maximum weight
of satisfiable equations in a {\sc Max Lin} system $Az=b$. Indeed, $W/2$ is the average weight of satisfied equations (as the probability of each equation to be satisfied is $1/2$) and, thus, is a lower bound; to see
the tightness consider a system of pairs of equations of the
form $\sum_{i\in I}z_i=0,\ \sum_{i\in I}z_i=1$ of weight 1.
Mahajan et al.
\cite{MahajanRamanSikdar06,MahajanRamanSikdar09} parameterized \textsc{Max
Lin} as follows: given a {\sc Max Lin} system $Az=b$, decide whether
the total weight of satisfied equations minus $W/2$ is at least $k'$,
where $W$ is the total weight of all equations and $k'$ is the parameter.
This is equivalent to asking whether the maximum excess is at least $k$,
where $k=2k'$ is the parameter.
(Note that since $k=2k'$, these two questions are equivalent from the
complexity point of view.)
Since $W/2$ is the average weight of satisfied equations, we will
call the parameterized \textsc{Max Lin} problem
\textsc{Max Lin Above Average} or \textsc{Max Lin AA}. 
Since the parameter $k$ is more convenient for us to use, in what follows we
use the version of \textsc{Max Lin AA} parameterized by $k.$

Mahajan et al. \cite{MahajanRamanSikdar06,MahajanRamanSikdar09} raised the question of determining the parameterized complexity of \textsc{Max Lin AA}.
It is not hard to see (we explain it in detail in Section \ref{CSPsec}) that we may assume that no two equations in $Az=b$ have the same left-hand side
and $n={\rm rank A}$. Using our maximum excess results, we prove that, under these assumptions, (a) \textsc{Max Lin AA} is fixed-parameter tractable if $m\le 2^{p(n)}$
for an arbitrary fixed function $p(n)=o(n)$, and (b) \textsc{Max Lin AA} has a polynomial-size kernel if $m\le 2^{n^a}$ for an arbitrary $a<1$.
We conjecture that under the two assumptions if $m<2^{an}$ for some constant $a>0$, then \textsc{Max Lin AA} is W[1]-hard, i.e., result (a) is best possible in a sense.
In addition, we prove that \textsc{Max Lin AA} is in XP (thus, \textsc{Max Lin AA} is polynomial-time solvable for every fixed $k$), and, moreover, it is in W[P].

Recall that {\sc Max $r$-Lin AA} is a special case of {\sc Max Lin AA}, where each equation has at most $r$ variables.
In {\sc Max Exact $r$-SAT AA} we are given a multiset of $m$ clauses on $n$ variables such that each clause has $r$ variables and asked whether there is a truth assignment to the $n$ variables that
satisfies at least $(1-2^{-r})m + k2^{-r}$ clauses.
Using our maximum excess results, we prove that for each fixed $r\ge 2$ Max $r$-Lin AA has a kernel with $O(k \log k)$ variables and, thus, it can be solved in time $2^{O(k \log k)}+m^{O(1)}.$ This improves a kernel with $O(k^2)$ variables for Max $r$-Lin AA obtained by Gutin et al. \cite{GutinKimSzeiderYeo09a}. Similarly, we prove that for each $r\ge 2$ {\sc Max Exact $r$-SAT AA} has a kernel with $O(k \log k)$ variables and it
can be solved in time $2^{O(k \log k)}+m^{O(1)}$ improving a kernel with $O(k^2)$ variables for {\sc Max Exact $r$-SAT AA} obtained by Alon et al. \cite{AlonSODA2010}. Note that while the kernels with $O(k^2)$ variables were obtained using a probabilistic approach, our results are obtained using an algebraic approach. Using a graph-theoretical approach Alon et al. \cite{AlonSODA2010} obtained a kernel of {\sc Max Exact 2-SAT AA} with $O(k)$ variables, but it is unlikely that their approach can be extended beyond $r=2$.

Fourier analysis of pseudo-boolean functions, i.e., functions $f:\ \{-1,+1\}^n\rightarrow \mathbb{R}$, has been used in many areas of computer science(cf. \cite{AlonSODA2010,odonn,deWolf}). In Fourier analysis, the Boolean domain is often assumed to be $\{-1,+1\}^n$ rather than more usual $\{0,1\}^n$ and we will follow this assumption in our paper. Here we use the following well-known and easy to prove fact \cite{odonn} that each function $f:\ \{-1,+1\}^n\rightarrow \mathbb{R}$ can be uniquely written as
\begin{equation}\label{eq1}f(x)=\sum_{S\subseteq [n]}c_S\prod_{i\in S}x_i,\end{equation}
where $[n]=\{1,2,\ldots ,n\}$ and each $c_S$ is a real. Formula (\ref{eq1}) is the Fourier expansion $f$, $c_S$ are the Fourier coefficients of $f$ (in the literature, these coefficients are often denoted by $\hat{f}(S)$ and we will use both $c_S$ and $\hat{f}(S)$ interchangeably), and the monomials $\prod_{i\in S}x_i$ form an orthogonal basis of (\ref{eq1}) (thus, the monomials are often written as $\chi_S(x)$ but we will use only $\prod_{i\in S}x_i$ as it is more transparent).

Optimization of pseudo-boolean functions is useful in many areas including computer science, discrete mathematics, operations research,  statistical mechanics and manufacturing;
for many results and applications of pseudo-boolean function optimization, see a well-cited survey \cite{BH02}.
In classical analysis, there is a large number of lower bounds on the maxima of trigonometric Fourier expansions, cf. \cite{Borwein}.
In Section \ref{LBsec}, we prove a sharp lower bound on the maximum of a pseudo-boolean function using its Fourier expansion. The bound can be used in algorithmics, e.g., for
approximation algorithms.

\section{Results on Maximum Excess}\label{CSPsec}

Consider two reduction rules for {\sc Max Lin} introduced in \cite{GutinKimSzeiderYeo09a} for {\MLT}.
These rules are of interest due to Lemma \ref{mGEn}.

\begin{krule}\label{rulerank}
Let $t={\rm rank} A$ and let columns $a^{i_1},\ldots ,a^{i_t}$ of $A$ be linearly independent.
Then delete all variables not in $\{z_{i_1},\ldots ,z_{i_t}\}$ from the equations of $Az=b$.
\end{krule}

\begin{krule}\label{rule1}
If we have, for a subset $S$ of $[n]$, an equation $\sum_{i \in S} z_i =b'$
with weight $w'$, and an equation $\sum_{i \in S} z_i =b''$ with weight $w''$,
then we replace this pair by one of these equations with weight $w'+w''$ if $b'=b''$ and, otherwise, by
the equation whose weight is bigger, modifying its
new weight to be the difference of the two old ones. If the resulting weight
is~0, we delete the equation from the system.
\end{krule}

\begin{lemma}\label{mGEn}
Let $A'z=b'$ be obtained from $Az=b$ by Rule~\ref{rulerank} or \ref{rule1}.
Then the maximum excess of $A'z=b'$  is equal to the maximum excess of $Az=b$.
Moreover, $A'z=b'$ can be obtained from $Az=b$ in time polynomial in $n$ and $m$.
\end{lemma}

To see the validity of Rule~\ref{rulerank}, consider an independent set $I$ of columns of $A$ of cardinality ${\rm rank} A$ and
a column $a^{j}\not\in I$. Observe that $a^j=\sum_{i\in I'}a^i,$ where $I'\subseteq I$. 
Consider an assignment $z=z^0$.
If $z^0_j=1$ then for each $i\in I'\cup \{j\}$ replace $z^0_i$ by $z^0_i+1$.
The new assignment satisfies exactly the same equations as the initial assignment. Thus, we may assume that $z_j = 0$ and remove $z_j$ from the system.
For a different proof, see \cite{GutinKimSzeiderYeo09a}.
If we cannot change a weighted system $Az=b$ using Rules~\ref{rulerank} and \ref{rule1}, we call it {\em irreducible}.

\vspace{3mm}

Consider the following algorithm that tries to maximize the total weight of satisfied equations of $Az=b$.
We assume that, in the beginning, no equation or variable in $Az=b$ is marked.

\begin{center}
\fbox{~\begin{minipage}{11cm}
\textsc{Algorithm $\cal H$}

\smallskip
While the system $Az=b$ is nonempty do the following:

\begin{enumerate}

					                   \item Choose an arbitrary equation  $\sum_{i \in S} z_i =b$ and mark $z_l$, where $l= \min\{i:\ i\in S\}.$
                                       \item Mark this equation and delete it from the system.
                                       \item Replace every equation $\sum_{i \in S'} z_i =b'$ in the system containing $z_l$ by
                                      $\sum_{i \in S} z_i + \sum_{i \in S'} z_i = b+b'$.
                                       \item Apply Reduction Rule \ref{rule1} to the system.
                                     \end{enumerate}
\smallskip
\end{minipage}~}
\end{center}

Note that algorithm ${\cal H}$ replaces $Az=b$ with an \emph{equivalent}
system under the assumption that the marked equations are satisfied; that
is, for every assignment of values to the variables $z_1, \dots, z_n$ that
satisfies the marked equations, both systems have the same excess.

The {\em maximum ${\cal H}$-excess} of $Az=b$ is the maximum possible total weight of equations marked by ${\cal H}$ for $Az=b$
taken over all possible choices in Step 1 of $\cal H$.

\begin{lemma}\label{lemH}
The maximum excess of $Az=b$ equals its maximum ${\cal H}$-excess.
\end{lemma}
\begin{proof}
We first prove that the maximum excess of $Az=b$ is not smaller than its maximum ${\cal H}$-excess.

Let $K$ be the set of equations marked by $\cal H$.
A method first described in \cite{CGJ} can find an assignment of values to the variables such that the equations in $K$ are satisfied and, in the remainder of the system, the total weight of satisfied equations is not smaller than the total weight of falsified equations.

For the sake of completeness, we repeat the description here. By construction, for any assignment that satisfies all the marked equations, exactly half of the non-marked equations are satisfied. Therefore it suffices to find an assignment to the variables such that all marked equations are satisfied. This is possible if we
find an assignment that satisfies the last marked equation, then find an assignment satisfying the equation marked before the last, etc. Indeed,
the equation marked before the last contains
a (marked) variable $z_l$ not appearing in the last equation, etc. This proves the first part of our lemma.

Now we prove that the maximum ${\cal H}$-excess of $Az=b$ is not smaller than its maximum excess. Let $z = (z_1,\ldots ,z_n)$ be an assignment that achieves the maximum excess, $t$.
Observe that if at each iteration of $\cal H$ we mark an equation that is satisfied by $z$, then $\cal H$ will mark equations of total weight $t$.
\end{proof}

\begin{remark}\label{rem1}
{\rm It follows from Lemma \ref{lemH} that the maximum excess of a (nonempty) irreducible system $Az=b$ with smallest weight $w_{\min}$ is at least $w_{\min}$. If all weights are integral, then the maximum excess of  $Az=b$ is at least 1.}
\end{remark}

Clearly, the total weight of equations marked by $\cal H$ depends on the choice of equations to mark in Step 1.
Below we consider one such choice based on the following theorem. The theorem allows us to find a set of equations such that we can mark each equation in the set in successive iterations of $\cal H$. This means we can run $\cal H$ a guaranteed number of times, which we can use to get a lower bound on the ${\cal H}$-excess.

\begin{theorem}\label{thmKcon}
Let $M$ be a set in $\mathbb{F}^n_2$ such that $M$ contains a basis of $\mathbb{F}^n_2$, the
zero vector is in $M$ and $|M|<2^n$. If $k$ is a positive integer and $k+1\le |M|\le 2^{n/k}$
then, in time $|M|^{O(1)}$, we can find a subset $K$ of $M$ of $k+1$ vectors
such that no sum of two or more vectors of $K$ is in $M$.
\end{theorem}

\begin{proof}
We first consider the case when $k = 1$. Since $|M|<2^n$ and the zero vector is in $M$, there is a non-zero vector
$v\not\in M$. Since $M$ contains a basis for $\mathbb{F}^n_2$, $v$ can be written as a sum
of vectors in $M$ and consider such a sum with the minimum number of summands: $v=u_1+\cdots + u_{\ell}$, $\ell \ge 2$.  Since $u_1+u_2\not\in M$, we may set $K=\{u_1,u_2\}.$
We can find such a set $K$ in polynomial time by looking at every pair in $M\times M$.

We now assume that $k>1$. Since $k+1\le |M|\le 2^{n/k}$ we have $n\ge k+1.$

We proceed with a greedy algorithm that tries to find $K$. Suppose we have a set $L = \{a_1,\ldots ,a_l\}$ of vectors in $M$, $l \le k$,  such that no sum of two or more elements of $L$ is in $M$.  We can extend this set to a basis, so $a_1=(1,0,0,\ldots ,0),\ a_2=(0,1,0,\ldots ,0)$ and so on. For every $a \in M \backslash L$ we check whether $M \backslash \{a_1,\ldots, a_l, a\}$ has an element that agrees with $a$ in all co-ordinates $l+1, \ldots, n$. If no such element exists, then we add $a$ to the set $L$, as no element in $M$ can be expressed as a sum of $a$ and a subset of $L$.

If our greedy algorithm finds a set $L$ of size at least $k+1$, we are done and $L$ is our set $K$. Otherwise, we have stopped at $l \le k$. In this case,
we do the next iteration as follows. Recall that $L$ is part of a basis of $M$ such that $a_1 = (1,0,0, \ldots, 0),\ a_2 = (0,1,0, \dots , 0), \ldots .$
We  create a new set $M'$ in $\mathbb{F}^{n'}_2$, where $n' = n-l$. We do this\footnote{For the reader familiar with vector space terminology: $\mathbb{F}^{n'}_2$ is $\mathbb{F}^n_2$ modulo ${\rm span}(L)$, the subspace of $\mathbb{F}^n_2$
spanned by $L$, and $M'$ is the image of $M$ in $\mathbb{F}^{n'}_2$.} by removing the first $l$ co-ordinates from $M$, and then identifying together any vectors that agree in the remaining $n'$ co-ordinates. We are in effect identifying together any vectors that only differ by a sum of some elements in $L$. It follows that every element of $M'$ was created by identifying together at least two elements of $M$, since otherwise we would have had an element in $M\backslash L$ that should have been added to $L$ by our greedy algorithm. Therefore it follows that $|M'| \le |M|/2 \le 2^{n/k -1}$.
From this inequality and the fact that $n' \ge n - k$, we get that $|M'| \le 2^{n'/k}.$
It also follows by construction of $M'$ that $M'$ has a basis for $\mathbb{F}^{n'}_2$, and that the zero vector is in $M'$. (Thus, we have $|M'|\ge n'+1$.)
If $n'\ge k+1$ we complete this iteration by running the algorithm on the set $M'$ as in the first iteration.
Otherwise ($n'\le k$), the algorithm stops.

Since each iteration of the algorithm decreases $n'$, the algorithm terminates.
Now we prove that at some iteration, the algorithm will actually find a set $K$ of $k+1$ vectors.
To show this it suffices to prove that we will never reach the point when $n' \le k$.
Suppose this is not true and we obtained $n' \le k$.
Observe that $n'\geq 1$
(before that we had $n'\geq k+1$  and we decreased $n'$ by at most $k$) and $|M'|\geq n'+1$. Since $|M'|\leq 2^{n'/k}$, we have $n'+1\leq
2^{n'/k}$, which is impossible due to $n'\leq k$ unless $n'=1$ and $k=1$, a
contradiction with the assumption that $k>1$.

It is easy to check that the running time of the algorithm is polynomial in $|M|.$
\end{proof}

\begin{remark} {\rm It is much easier to prove a non-constructive version of the above result. In fact we can give a non-constructive proof that $k+1\le |M|\le 2^{n/k}$ can be replaced by $2k<|M|<2^{n/k}((k-1)!)^{1/k}$. We will extend our proof above for the case $k=1$. We may assume that $k\ge 2$.
Observe that the number of vectors of $\mathbb{F}^n_2$ that can be expressed as the sum of at most $k$ vectors of $M$ is at most
$$\binom{|M|}{k} + \binom{|M|}{k-1} + \dots + \binom{|M|}{1} + 1 \le {|M|^k}/{(k-1)!} \text{ for }|M|> 2k.$$

Since $|M|<2^{n/k}((k-1)!)^{1/k}$ we have $|\mathbb{F}^n_2|>{|M|^k}/{(k-1)!}$ and, thus, at least for one vector $a$ of $\mathbb{F}^n_2$ we have $a=m_1+\cdots + m_{\ell}$, where $\ell$ is minimum and $\ell>k.$ Note that, by the minimality of $\ell$, no sum of two or more summands of the sum for $a$ is in $M$ and all summands are distinct. Thus, we can set $K=\{m_1,\ldots ,m_{k+1}\}.$
}
\end{remark}


\begin{theorem}\label{lemYes}
Let $Az=b$ be an irreducible system, let $w_{\min}$ be the minimum weight of an equation in $Az=b$, and let $k\ge 2.$ If $k\le m\le 2^{n/(k-1)}-2$, then the maximum excess of $Az=b$ is at least $k\cdot w_{\min}$. Moreover, we can find an assignment that achieves an excess of at least $k\cdot w_{\min}$ in time $m^{O(1)}$.
\end{theorem}
\begin{proof}
Consider a set $M$ of vectors in $\mathbb{F}_2^n$ corresponding to equations in $Az=b$ as follows: for each $\sum_{i \in S} z_i =b_S$ in $Az=b$, the vector $v=(v_1,\ldots ,v_{n})\in M$, where $v_i=1$ if $i\in S$ and $v_i=0$, otherwise.
Add the zero vector to $M$. As $Az=b$ is reduced by Rule \ref{rulerank} and $k\le m\le 2^{n/(k-1)}-2$, we have that $M$ contains a basis for $\mathbb{F}^n_2$ and $k\le |M| \le 2^{n/(k-1)}-1$. Therefore, using Theorem \ref{thmKcon} we can find a set $K$ of $k$ vectors such that no sum of two or more vectors in $K$ belongs to $M.$

Now run Algorithm $\cal H$ choosing at each Step 1 an equation of $Az=b$  corresponding to a member of $K$, then equations picked at random until the algorithm terminates.
Algorithm $\cal H$ will run at least $k$ iterations as no equation corresponding to a vector in $K$ will be deleted before it has been marked. Indeed, suppose that this is not true. Then there are vectors $w\in K$ and $v\in M$ and a pair of nonintersecting subsets $K'$ and $K''$ of $K\setminus \{v,w\}$ such that $w+\sum_{u\in K'}u=v+\sum_{u\in K''}u$. Thus,
$v=w+\sum_{u\in K'\cup K''}u$, a contradiction with the definition of $K.$

In fact, the above argument shows that no equation of $Az=b$  corresponding to a member of $K$ will change its weight during the first $k$ iterations of $\cal H$.
Thus, by Lemma \ref{lemH}, the maximum excess of $Az=b$ is at least $k\cdot w_{\min}$.
It remains to observe that we can once again use the algorithm given in the proof of Lemma \ref{lemH} to find an assignment that gives an excess of at least $k\cdot w_{\min}$.
\end{proof}

We now provide a useful association between  weighted systems of linear equations on $\mathbb{F}^n_2$ and Fourier expansions of functions $f:\ \{-1,+1\} \rightarrow \mathbb{R}$. Let us rewrite (\ref{eq1}), the Fourier expansion of such a function, as
\begin{equation}\label{eq2}f(x)=\hat{f}(\emptyset)+\sum_{S\in {\cal F}}c_S\prod_{i\in S}x_i,\end{equation} where ${\cal F}=\{\emptyset\neq S\subseteq [n]:\ c_S\neq 0\}.$

Now associate the polynomial $\sum_{S\in {\cal F}}c_S\prod_{i\in S}x_i$ in (\ref{eq2}) with a weighted system $Az=b$ of linear equations on $\mathbb{F}^n_2$: for each $S\in {\cal F}$, we have an equation $\sum_{i\in S}z_i=b_S$ with weight $|c_S|$, where $b_S=0$ if $c_S$ is positive and $b_S=1$, otherwise. Conversely, suppose we have a system $Az=b$ of linear equations on $\mathbb{F}^n_2$ in which each equation $\sum_{i \in S} z_i =b_S$ is assigned a weight $w_S>0$ and no pair of equations have the same left-hand side. This system can be associated with the polynomial $\sum_{S\in {\cal F}}c_S\prod_{i\in S}x_i,$ where  $c_S=w_S$, if $b_S=0$,  and $c_S=-w_S$, otherwise. The above associations provide a bijection between Fourier expansions of functions $f:\ \{-1,+1\} \rightarrow \mathbb{R}$ with $\hat{f}(\emptyset)=0$ and weighted systems of linear equations on $\mathbb{F}^n_2$. This bijection is of
interest due to the following:

\begin{proposition}\label{lemFS}
An assignment $z^{(0)}=(z^{(0)}_1,\ldots ,z^{(0)}_n)$ of values to the variables of $Az=b$ maximizes the total weight of satisfied equations of $Az=b$ if and only if
$x^{(0)}=((-1)^{z^{(0)}_1},\ldots ,(-1)^{z^{(0)}_n})$ maximizes $f(x).$ Moreover, $\max_{x\in \{-1,+1\}^n}f(x)- \hat{f}(\emptyset)$ equals the maximum excess of $Az=b.$
\end{proposition}
\begin{proof}
The claims of this lemma easily follow from the fact
that an equation $\sum_{i\in S}z_i=0$ is satisfied if and only if $\prod_{i\in S}x_i>0,$ where $x_i=(-1)^{z_i}.$
\end{proof}

\section{Corollaries} \label{LBsec}

This section contains a collection of corollaries of Theorem \ref{lemYes} establishing parameterized complexity of special cases of {\sc Max Lin AA}, of {\sc Max Exact $r$-SAT}, and of a
wide class of constraint satisfaction problems. In addition, we will prove that {\sc Max Lin AA} is in X[P] and obtain a sharp lower bound on the maximum of a pseudo-boolean function.

\subsection{Parameterized Complexity of Max Lin AA}

\begin{corollary}\label{fptT}
Let $p(n)$ be a fixed function such that $p(n)=o(n)$. If $m\le 2^{p(n)}$ then {\MLT} is fixed-parameter tractable.
Moreover, a satisfying assignment can be found in time $g(k)m^{O(1)}$ for some computable function $g$.
\end{corollary}
\begin{proof}
We may assume that $m \ge n > k > 1.$ Observe that $m\le 2^{n/k}$ implies $m\le 2^{n/(k-1)}-2$. Thus, by Theorem \ref{lemYes}, if $p(n)\le n/k$, the answer to {\MLT} is {\sc yes}, and there is a polynomial algorithm to find a suitable assignment. Otherwise, $n\le f(k)$ for some function dependent on $k$
only and {\MLT} can be solved in time $2^{f(k)}m^{O(1)}$ by checking every possible assignment.
\end{proof}

Let $\rho_i$ be the number of equations in  $Az=b$ containing $z_i$, $i=1,\ldots ,n$. Let $\rho=\max_{i\in [n]}\rho_i$ and
let $r$ be the maximum number of variables in an equation of $Az=b$.
Crowston et al. \cite{CGJ} proved that {\MLT} is fixed-parameter tractable if either $r\le r(n)$ for some fixed function $r(n)=o(n)$ or $\rho\le \rho(m)$ for some fixed function $\rho(m)=o(m)$.

For a given $r=r(n)$, we have $m\le \sum_{i=1}^r{n \choose i}$. By Corollary 23.6 in \cite{jukna}, $m\le 2^{nH(r/n)}$, where $H(y)=-y \log_2 y - (1-y) \log_2 (1- y)$, the entropy of $y.$ It is easy to see that if $y=o(n)/n$, then $H(y)=o(n)/n.$ Hence, if $r(n)=o(n)$, then $m\le 2^{o(n)}.$
By Corollary 23.5 in \cite{jukna} (this result was first proved by Kleitman et al. \cite{kleitman}), for a given $\rho=\rho(m)$ we have
$m\le 2^{nH(\rho/m)}.$  Therefore, if $\rho(m)=o(m)$ then $m\le 2^{n\cdot o(m)/m}$ and, thus, $m\le 2^{o(n)}$ (as $n\le m$, if $n\rightarrow \infty$ then $m\rightarrow \infty$ and $o(m)/m\rightarrow 0$).
Thus, both results of Crowston et al. \cite{CGJ} follow from corollary \ref{fptT}.

\vspace{3mm}

Similarly to Corollary \ref{fptT} it is easy to prove the following:

\begin{corollary}
Let $0<a<1$ be a constant. If $m<2^{O(n^{a})}$ then {\MLT} has a kernel with $O(k^{1/(1-a)})$ variables.
\end{corollary}

By Corollary \ref{fptT} it is easy to show that {\MLT} is in XP.

\begin{proposition}
{\MLT} can be solved in time $O(m^{k+O(1)}).$
\end{proposition}
\begin{proof}
We may again assume $m \ge n > k>1.$ As in the proof of Corollary \ref{fptT}, if
$m\le 2^{n/k}$ then the answer to {\MLT} is {\sc yes}
and a solution can be found in time $m^{O(1)}$. Otherwise,
$2^n<m^k$ and {\MLT} can be solved in time $O(m^{k+2}).$
\end{proof}

In fact, it is possible to improve this result, as the next theorem shows.

\begin{theorem} \label{thmWP}
{\MLT} is in W[P].
\end{theorem}

To prove this theorem we make use of the following lemma from \cite{FlumGrohe06} (Lemma 3.8, p. 48). Here $k(x)$ is the value of the parameter on an instance $x\in \Sigma^{\ast}.$

\begin{lemma} \label{lemWP}
A parameterized problem $(Q,k)$ over the alphabet $\Sigma$ is in W[P] if and only if there are computable functions $f,h: \mathbb{N} \rightarrow \mathbb{N}$, a polynomial $p(X)$, and a $Y \subseteq \Sigma^{\ast} \times \{ 0,1 \}^{\ast}$ such that:
\begin{description}
  \item[(i)] For all $(x,y) \in \Sigma^{\ast}\times \{ 0,1 \}^{\ast}$, it is decidable in time $f(k(x))\cdot p(|x|)$ whether $(x,y) \in Y$.
\item[(ii)] For all $(x,y) \in \Sigma^{\ast}\times \{ 0,1 \}^{\ast}$, if $(x,y) \in Y$ then $|y| = h(k(x)) \cdot \lfloor \log_2|x| \rfloor.$
\item[(iii)] For every $x \in \Sigma^{\ast}$
$$x \in Q \Longleftrightarrow  \text{there exists a } y \in \{ 0,1 \}^{\ast} \text{such that } (x,y) \in Y.$$
\end{description}
\end{lemma}

\begin{proof}[Proof of Theorem \ref{thmWP}]
Recall from Lemma \ref{lemH} that the maximum excess of $Az=b$ is at least
$k$  if and only if
we can run algorithm $\cal H$ a number of times and get a total weight of marked equations at least $k$.

Suppose we are given a sequence $e_1, \ldots, e_l$ of equations to mark in each
iteration of $\cal H$. We can, at the $i$'th iteration of $\cal H$, mark equation $e_i$ as
long as $e_i$ is still in the system. If we are able to mark all the
equations $e_1, \ldots e_l$, we can then check that the total weight of these
marked equations is at least $k$. If it is, then we know we have a
{\sc yes}-instance.
Conversely, if the system has a maximum excess of at least $k$, then there will be some sequence $e_1, \ldots, e_l$ that gives us a total weight of marked equations at least $k$.
Furthermore, by integrality of the weights, we may assume that $l \le k$.
We use this idea to construct a set $Y$ that satisfies the conditions of Lemma \ref{lemWP}.

Firstly we show that a sequence of $l \le k$ equations can be encoded as a string $y \in \{ 0,1 \}^{\ast}$ of length $2k \cdot \lfloor \log_2|x| \rfloor$, where $x$ is an instance of {\MLT}. Let the equations be numbered from $1$ to $m$, then we can express a sequence of equations $e_1, \ldots e_l$, as a sequence of $k$ integers between $0$ and $m$ (if $l < k$ then we end the sequence with $k-l$ zeroes). Each integer between $0$ and $m$ can be expressed by a string in $\{ 0,1 \}^{\ast}$ of length at most $\lceil \log_2m \rceil \le \lceil \log_2|x| \rceil$, so certainly it can be expressed by a string of length $2 \lfloor \log_2|x| \rfloor$. Therefore we can express the $k$ integers as a string of length $2k \cdot \lfloor \log_2|x| \rfloor$.

For an instance $x$ of {\MLT} and a string $y \in \{ 0,1 \}^{\ast}$, let us call $y$ a \emph{certificate for $x$} if $|y| = 2k \cdot \lfloor \log_2|x| \rfloor$ and $y$ encodes a sequence of $k$ integers corresponding to a sequence of equations $e_1, \ldots, e_l$ in $x$, such that by marking each equation in turn in iterations of $\cal H$, we get a set of marked equations of weight at least $k$.
It follows that $x$ is a {\sc yes}-instance if and only if there exists a certificate for $x$. Furthermore we can check in polynomial time whether $y$ is a certificate of $x$ by trying to convert $y$ into a sequence of equations and running algorithm $\cal H$ marking those equations. (This is in fact a stronger result than we require for this proof - we only need that the algorithm is fixed-parameter tractable rather than polynomial.)

We now let
$$Y = \{ (x,y) \in \Sigma^{\ast}\times \{ 0,1 \}^{\ast} | x \text{ is a {\sc yes}-instance of {\MLT} and } y \text{ is a certificate of } x \}$$
and let $Q$ be the set of all {\sc yes}-instances of {\MLT}.
By definition of $Y$ and the definition of a certificate, conditions (ii) and (iii) of Lemma \ref{lemWP} are satisfied.
As we can determine in polynomial time whether $y$ is a certificate for $x$, condition (i) is also satisfied. Therefore, by Lemma \ref{lemWP}, {\MLT} is in W[P].

\end{proof}

\subsection{Max $r$-Lin AA, Max Exact $r$-SAT AA and Max $r$-CSP AA}

Using Theorem \ref{lemYes} we can prove the following two results.

\begin{corollary}\label{cor1}
Let $r\ge 2$ be a fixed integer. Then {\sc Max $r$-Lin AA} has a kernel with $O(k \log k)$ variables and can be solved in time $2^{O(k \log k)} + m^{O(1)}$.
\end{corollary}
\begin{proof}
Observe that $m\le n^r$ and $n^r\le 2^{n/(k-1)}-2$ if $n\ge c(r)k\log_2 k$ provided $c(r)$ is large enough ($c(r)$ depends only on $r$). Thus, by Theorem \ref{lemYes}, if $n\ge c(r)k\log_2 k$ then the answer to {\sc Max $r$-Lin AA} is {\sc yes}. Hence, we obtain a problem kernel with at most $c(r)k\log_2 k=O(k \log k)$ variables and, therefore, can solve {\sc Max $r$-Lin AA} in time $2^{O(k \log k)} + m^{O(1)}$.
\end{proof}

\begin{corollary}\label{cor2}
Let $r\ge 2$ be a fixed integer. Then there is a bikernel from {\sc Max Exact $r$-SAT} to {\sc Max $r$-Lin AA} with $O(k \log k)$ variables. Moreover, {\sc Max Exact $r$-SAT} has a kernel with $O(k \log k)$ variables and can be solved in time $2^{O(k \log k)} + m^{O(1)}$.
\end{corollary}
\begin{proof}
Let $F$ be an $r$-CNF formula with clauses $C_1,\ldots ,C_m$
in the variables $x_1, x_2, \ldots ,x_n$. We may assume that $x_i \in \{-1,1\}$, where $-1$ corresponds to {\sc true}.
For $F$, following \cite{AlonSODA2010} consider
\[
g(x)=\sum_{C\in F}[1-\prod_{x_i\in {\rm var}(C)}(1+\epsilon_ix_i)],
\]
where ${\rm var}(C)$ is the set of variables of $C$, $\epsilon_i\in \{-1,1\}$ and $\epsilon_i=1$ if and only if $x_i$
is in~$C$. It is shown in \cite{AlonSODA2010} that the answer to {\sc Max Exact $r$-SAT} is {\sc yes} if and only if there is a truth assignment $x^0$ such that $g(x^0)\ge k.$

Algebraic simplification of $g(x)$ will lead us to
Fourier expansion of $g(x)$: \begin{equation}g(x)=\sum_{S\in {\cal F}}c_S\prod_{i\in S}x_i,\end{equation} where ${\cal F}=\{\emptyset\neq S\subseteq [n]:\ c_S\neq 0, |S|\le r \}$. Thus, $|{\cal F}|\le n^r$. By Proposition \ref{lemFS}, $\sum_{S\in {\cal F}}c_S\prod_{i\in S}x_i$ can be viewed as an instance of {\sc Max $r$-Lin} and, thus, we can reduce {\sc Max Exact $r$-SAT}
into {\sc Max $r$-Lin} in polynomial time (the algebraic simplification can be done in polynomial time as $r$ is fixed). By Corollary  \ref{cor1}, {\sc Max $r$-Lin} has a kernel with $O(k \log k)$ variables. This kernel is a bikernel from {\sc Max Exact $r$-SAT} to {\sc Max $r$-Lin}. Using this bikernel, we can solve {\sc Max Exact $r$-SAT} in time $2^{O(k \log k)} + m^{O(1)}$.

It remains to use the transformation described in \cite{AlonSODA2010} of a bikernel from {\sc Max Exact $r$-SAT} to {\sc Max $r$-Lin} into a kernel of {\sc Max Exact $r$-SAT}. This transformation gives us a kernel with $O(k \log k)$ variables.
\end{proof}

In the Boolean Max-$r$-Constraint Satisfaction Problem ({\sc Max-$r$-CSP}), we are given a collection of Boolean functions, each involving at most $r$ variables, and asked to find a truth assignment that satisfies as many functions as possible. We will consider the following parameterized version of {\sc Max-$r$-CSP}. We are given a set $\Phi$ of Boolean functions, each involving at
most $r$ variables, and a collection ${\cal F}$ of $m$ Boolean functions, each $f \in \cal F$ being a
member of $\Phi$, and each acting on some subset of
the $n$ Boolean variables $x_1,x_2, \ldots ,x_n$ (each $x_i\in \{-1,1\}$). We are to decide whether there is a truth assignment to the $n$ variables such that the total number of satisfied functions is at least $E+k2^{-r}$, where  $E$ is the average value of the number of satisfied functions.

\begin{corollary}\label{cor3}
Let $r\ge 2$ be a fixed integer. Then there is a bikernel from {\sc Max $r$-CSP} to {\sc Max $r$-Lin AA} with $O(k \log k)$ variables. {\sc Max $r$-CSP} can be solved in time $2^{O(k \log k)} + m^{O(1)}$.
\end{corollary}
\begin{proof}
Following \cite{AGK04} for a boolean
function $f$ of $r(f)\le r$ boolean variables
$
x_{i_1},  \ldots , x_{i_{r(f)}},
$ introduce
a polynomial $h_f(x),\ x=(x_1, x_2, \ldots ,x_n)$ as follows. Let $V_f \subset \{-1,1\}^{r(f)}$ denote the set of all
satisfying assignments of $f$. Then
$$
h_f(x) = 2^{r-r(f)}\sum_{(v_1, \ldots ,v_{r(f)}) \in V_f}
[\prod_{j=1}^{r(f)} (1+x_{i_j} v_j) - 1].
$$

Let $h(x)=\sum_{f \in \cal F} h_f(x).$
It is easy to see (cf. \cite{AlonSODA2010}) that the value of $h(x)$ at $x^0$ is precisely $2^r(s-E)$, where $s$ is the number of
the functions satisfied by the truth assignment $x^0$, and $E$ is the
average value of the number of satisfied functions.  Thus, the answer to {\sc Max-$r$-CSP} is {\sc yes} if and only if there is a truth assignment $x^0$ such that $h(x^0)\ge k.$
The rest of the proof is similar to that of Corollary \ref{cor2}.
\end{proof}

\subsection{Lower Bound on Maxima of Pseudo-boolean Functions}\label{PBsec}

\begin{corollary}\label{CLBT}
We have $\max_{x\in \{-1,+1\}^n}f(x) \ge \hat{f}(\emptyset) + (1+\lfloor \frac{{\rm rank} A}{\log_2 (|{\cal F}|+2)} \rfloor) \cdot \min_{S\in {\cal F}}|\hat{f}(S)|.$
\end{corollary}
\begin{proof}
Consider the system $Az=b$ associated with the Fourier expansion of $f$ according to the bijection described before Proposition \ref{lemFS}.
We may assume that the weighted system $Az=b$ has been simplified using Rule~\ref{rulerank} and, thus,
its number $n'$ of variables equals ${\rm rank} A$. Note that $n'\le m$, where $m$ is the number of equations in $Az=b$.
By Theorem \ref{lemYes}, Proposition \ref{lemFS} and the fact that $\min_{S\in {\cal F}}|\hat{f}(S)| = \min_jw_j$, it follows that if $k\le m\leq 2^{n'/(k-1)}-2$ then
 $$\max_{x\in \{-1,+1\}^n}f(x) - \hat{f}(\emptyset) \ge k\min_{S\in {\cal F}}|\hat{f}(S)|.$$

To complete the proof, recall that $n'={\rm rank} A$, $m=|{\cal F}|$ and observe that the maximum possible (integral) value of $k$  satisfying $m\leq 2^{n'/(k-1)}-2$ is $1+\lfloor \frac{{\rm rank} A}{\log_2 (|{\cal F}|+2)} \rfloor$.
\end{proof}

This bound is tight.
Indeed, consider the function $f(x)=-\sum_{\emptyset\neq S\subseteq [n]}\prod_{i\in S}x_i.$ Observe that $n={\rm rank} A$, $|{\cal F}|=2^n-1$ and, thus, $\max_{x\in \{-1,+1\}^n}f(x)\ge 1+\lfloor \frac{{\rm rank} A}{\log_2 (|{\cal F}|+2)} \rfloor = 1.$
If $x=(1,1,\ldots ,1)$ then $f(x)=-|{\cal F}|$ and if we set some $x_i=-1$ then after canceling out of monomials we see that $f(x)=1$. Therefore,  $\max_{x\in \{-1,+1\}^n}f(x)=1,$ and, thus, the bound of corollary \ref{CLBT} is tight. It is easy to see that the bound remains tight if we delete one monomial from $f(x)$. A sightly more complicated function showing that  the bound is tight is as follows: $g(x)=-\sum_{\emptyset\neq S\subseteq [n_1]}\prod_{i\in S}x_i - \sum_{S\in {\cal G}}\prod_{i\in S}x_i,$ where $n_1<n$ and ${\cal G}=\{S:\ \emptyset\neq S\subseteq [n], [n_1]\cap S = \emptyset\}.$

\begin{remark}
{\rm Consider {\sc Max Lin} with irreducible system $Az=b$ in which every equation is of weight 1. Then the bound of Theorem \ref{CLBT} gives an $(1/2 + (1+\delta)/m)$-approximation for {\sc Max Lin}, where $\delta=\lfloor n/\log_2 (m +2) \rfloor$. This is of interest since by the result of H{\aa}stad mentioned in Section \ref{section:intro}, $(1/2+\epsilon)$-approximation is impossible for any constant $\epsilon >0$ unless P=NP.
}\end{remark}

\vspace{3mm}

\medskip

\paragraph{Acknowledgments}
Gutin is thankful to Ilia Krasikov and Daniel Marx for discussions on the topic of the paper.
Research of Gutin, Jones and Kim was supported in part by an EPSRC
grant. Research of Gutin was also supported
in part by the IST Programme of the European Community, under the
PASCAL 2 Network of Excellence. Research of Ruzsa was supported by ERC--AdG
Grant No. 228005 and Hungarian National Foundation for Scientific
Research (OTKA), Grants No. 61908.

\urlstyle{rm}

\end{document}